\documentclass[letterpaper]{llncs}

\usepackage{amsmath,amssymb,amsfonts}
\usepackage{tgpagella,eulervm}
\usepackage{varioref}
\usepackage{enumerate}

\newcommand{\es}{\varnothing}

\DeclareMathOperator{\PSPACE}{\mathsf{PSPACE}}

\DeclareMathOperator{\mex}{\mathsf{mex}}

\begin{document}

\title{{\sc $P_3$-Games}}

\author{
Wing-Kai~Hon\inst{1}
\and 
Ton~Kloks\inst{}
\and 
Fu-Hong~Liu\inst{1}
\and 
Hsiang-Hsuan~Liu\inst{1,2}
\and
Tao-Ming~Wang\inst{3}
}

\institute{National Tsing Hua University, Hsinchu, Taiwan\\
{\tt (wkhon,fhliu,hhliu)@cs.nthu.edu.tw}
\and
University of Liverpool, Liverpool, United Kingdom\\
{\tt hhliu@liverpool.ac.uk}
\and
Tunghai University, Taichung, Taiwan\\
{\tt wang@go.thu.edu.tw}
}

\maketitle 

\begin{abstract}
Without further ado, we present the $P_3$-game. The $P_3$-game 
is decidable for elementary classes of graphs 
such as paths  and cycles. 
{F}rom an algorithmic point of 
view, the connected $P_3$-game is fascinating. We show that the 
connected $P_3$-game is polynomially decidable for classes 
such as trees, chordal graphs, ladders, cacti,  
outerplanar graphs and circular arc graphs.  
\end{abstract}

\section{Introduction}

\begin{definition}
An alignment is a pair $(X,\mathcal{L})$, where $X$ is a 
finite set and $\mathcal{L}$ is a collection of subsets of $X$ 
satisfying 
\begin{enumerate}
\item $\es \in \mathcal{L}$ and $X \in \mathcal{L}$, and 
\item $\mathcal{L}$ is closed under intersections. 
\end{enumerate}
\end{definition}
The elements of the set $\mathcal{L}$ of an alignment are called `closed.'

\begin{definition}
Let $(X,\mathcal{L})$ be an alignment and let $A \subseteq X$. 
The hull of $A$, denoted $\sigma(A)$, is defined as 
\[\sigma(A)=\bigcap_{K \in \mathcal{L} \;\text{\rm and}\; A \subseteq K} \; K.\] 
\end{definition}
The hull operator $\sigma$ is a closure operator, that is, it satisfies 
\begin{enumerate}[(i)]
\item $\sigma(\es)=\es$,
\item $S \subseteq \sigma(S)$ for each set $S$, 
\item For sets $S$ and $T$, if $S \subseteq T$ then $\sigma(S) \subseteq \sigma(T)$, and 
\item for each set $S$, $\sigma(\sigma(S))=\sigma(S)$. 
\end{enumerate}

\bigskip 

\begin{definition}
Let $G$ be a graph. A set $S \subseteq V(G)$ is $P_3$-closed if 
\[\forall_{x \in V(G)} \;\; x \notin S \quad\Rightarrow\quad |N(x) \cap S| < 2.\]
\end{definition}
For a graph $G$, the pair $(V(G),\mathcal{L})$, where 
$\mathcal{L}$ is the collection of $P_3$-closed sets in $G$, is an alignment. 

\bigskip 

Two players play the $P_3$-game on a graph by alternately selecting vertices. 
At the start of the game all vertices are unlabeled. During the game  
players label vertices. Prior to every move, the set of labeled 
vertices is $P_3$-closed. Let $L$ denote the set of labeled vertices. 
A move consists of labeling a, previously unlabeled, vertex $x$. 
In effect, the new set of labeled vertices 
becomes $\sigma(L+x)$.

\bigskip 

According to the Sprague-Grundy theory,  
when the $P_3$-game is played on a finite graph,    
there is a winning strategy for one of the two players.  
If there is a polynomial-time 
algorithm to decide whether there is a winning strategy 
for one of the two players, we call the game decidable (in polynomial time). 
For example, when the graph is a clique with 
at least two vertices, then the second player wins the game. That is so because 
the convex hull of any two vertices in a clique is $V(G)$. 

\bigskip 

Another example where the game is easy to decide is the case 
where the playground is a star. 

\begin{lemma}
Assume the graph is a star, $K_{1,t}$. Then player one has a 
winning strategy if and only if the number of leaves, $t$, is even. 
\end{lemma}
\begin{proof}
Assume the graph is a star with an even number of leaves. 
A winning move for player one is to choose the center. 

\medskip 

\noindent 
Assume the playground is a star with an odd number of leaves. 
The winning strategy for player two is to choose a vertex 
which leaves the playground with an even number of unlabeled 
leaves. 
\qed\end{proof}

\bigskip 

In his acclaimed 
paper on monophonic alignments, Duchet defines a {\em graph-alignment\/}  
as a pair $(G,\mathcal{C})$, where $G$ is a connected graph and $\mathcal{C}$ is a 
collection of subsets of $V(G)$ such that $(V(G),\mathcal{C})$ is an alignment 
and 
the following additional property holds. 

\[\boxed{\text{Every member of $\mathcal{C}$ induces a connected subgraph of $G$.}}\]

\bigskip 

The legal moves in the \underline{connected $P_3$-game} are restricted such that 
the $P_3$-closed sets induce connected subgraphs. Prior to every move, 
the playground is a connected graph $G$, 
with a set $L$ of labeled vertices satisfying 
\begin{enumerate}[\rm (1)]
\item $L$ is $P_3$-closed, and 
\item $G[L]$, that is the subgraph induced by L, is connected. 
\end{enumerate}
A legal move in the connected $P_3$-game is the selection of an unlabeled 
vertex $x$ such that $\sigma(L+x)$ induces a connected graph. 
In other words, a move is legal if the selected vertex is at 
distance at most two from the set $L$. 

\section{The connected $P_3$-game on trees}

To analyze a game played on graphs, one makes use of the `game graph.' 
This is a directed 
graph $(P,\Gamma)$, constructed as follows. 
Each node in the game graph represents a playground, 
and there is an arc $(u,v) \in \Gamma$ 
if the playground corresponding with the node $v$ 
can be reached from the playground corresponding with $u$ in one move. 

The Grundy function $g:P \rightarrow \mathbb{N} \cup \{0\}$ 
is defined on $P$ as follows.  If $p \in P$ is a sink, that is, a node without 
outgoing arcs, the Grundy value is defined as $g(p)=0$. For any other node, 
say $q$, the Grundy value is defined as 
\[g(q)=\min\;\{\; n\;|\; n \in \mathbb{N} \cup \{0\} 
\quad \text{and}\quad \forall_p \; (q,p) \in \Gamma 
\quad\Rightarrow\quad g(p) \neq n\;\}.\]
Thus, $g(q)$ is the smallest nonnegative integer which is not attained 
by any node in $P$ that can be reached from $q$ in one move. 
We allude that the Grundy function is sometimes called the $\mex$-function, 
which stands for {\em minimal excluded value\/}. 

Let $s \in P$ be the initial playground, before any move has been made. 
The Grundy value $g(s)$ is the Grundy value of the game $P$. 

The following theorem is easy to check. 

\begin{theorem}
\label{thm Grundy value}
Let $s \in P$ be the node representing the initial playground. Then player one 
has a winning strategy if and only if $g(s) \neq 0$. 
\end{theorem}
\begin{proof}
Notice that, from any node $q$ with $g(q) \neq 0$, a player can move 
to a node $p \in P$ with $g(p)=0$. 
\qed\end{proof}

\bigskip 

Consider a finite collection of games $G_1,\dots, G_t$. The product 
game 
\[G=G_1 \times \dots \times G_t\] 
is the game where a player is allowed to make 
a move in {\em one\/} of the games $G_i$.  

Let $g_i$ denote the Grundy value for game $G_i$. 
The nim-sum $g_1 \oplus \dots \oplus g_t$ is obtained as follows. 
Write each Grundy value $g_i$ in binary and add them up, 
{\em without a carry\/}. Thus, for example, $3\oplus 6=5$. 
Grundy's theorem is the following. 

\begin{theorem}[Grundy's Theorem]
\label{thm Grundy}
The Grundy value $g$ for the product game $G$ is 
\[g=g_1 \oplus \dots \oplus g_t.\] 
\end{theorem}

\bigskip 

\begin{theorem}
\label{thm connected P_3 on trees}
There exists a polynomial-time algorithm that decides if the first 
player has a winning strategy in the connected $P_3$-game played 
on a tree. 
\end{theorem}
\begin{proof}
Let $T$ be a tree. 
Consider a playground $p$, which is characterized by a 
connected subtree of $T$. 
A branch is a subtree with a maximal number of edges, which contains 
one node of the subtree $p$ as a leaf. We call the leaf 
the root of the branch. Let $B_1,\dots,B_{\ell}$ be the branches 
that have at least one edge. Denote the 
Grundy value of $B_i$ as $g_i$. The Grundy value of $p$ 
is the Grundy value of $B_1 \times \dots \times B_{\ell}$.  
By Grundy's Theorem~\ref{thm Grundy}, 
this is the nim-sum of the Grundy values $g_i$, 
that is, 
\begin{equation}
\label{equation Grundy}
g(p)=g_1 \oplus \dots \oplus g_{\ell}.
\end{equation}

\medskip 

\noindent 
For each possible initial move $s$ of the first player, the algorithm 
computes the Grundy value $g(s)$, via Equation~\eqref{equation Grundy}, 
by dynamic programming on the branches. 

\medskip 

\noindent 
By Theorem~\ref{thm Grundy value}, this proves the theorem. 
\qed\end{proof}

\section{The $P_3$-game on paths and cycles}

\subsubsection{Paths}

\begin{theorem}
\label{thm P_3 game on paths}
There exists an $O(n^2)$ algorithm that decides the $P_3$-game 
on paths. 
\end{theorem}
\begin{proof}
Consider the set of paths with $1$ up to $n$ vertices. 
Create a set in which 
each path occurs 4 times as it is bordered by a 
labeled vertex on either side or not.  So, the set contains 
$4n$ elements. For each element, the Grundy value is computed 
by considering all possible moves. Each move divides the path 
into at most two, strictly smaller elements. The algorithm 
processes the paths in the set in order of increasing length.
\qed\end{proof}

\begin{remark}
As for now, we don't have an easy formula which tells whether 
$P_n$ is won for player one. 
\end{remark}

\subsubsection{Cycles}

\begin{theorem}
There exists a polynomial-time algorithm to decide 
th $P_3$-game on cycles. 
\end{theorem}
\begin{proof}
When the cycle has an even number of vertices, the winning strategy for 
player two is to choose the vertex opposite player one's choice. 

\medskip 

\noindent 
The odd case is, unless it is $K_3$, less trivial. We can reduce it to paths and 
use Theorem~\ref{thm P_3 game on paths} as follows. Consider a 
cycle with an odd number $n > 3$ vertices. The first player selects 
some arbitrary vertex.  Split the selected vertex in two vertices, creating 
a path with two selected vertices at the ends, and $n-1$ unselected 
vertices between them. According to Theorem~\ref{thm P_3 game on paths}, there is an 
$O(n^2)$ algorithm to decide this game. 
\qed\end{proof}

\section{The connected $P_3$-game on paths and cycles}

\subsubsection{Cycles}

\begin{theorem}
The first player has a winning strategy when the game is played on a 
cycle $C_n$ if and only if $n \equiv 2 \pmod 3$. 
\end{theorem}
\begin{proof}
We show how to calculate the Grundy value. 
Consider a cycle with $n$ vertices. Let $f(k)$ be the Grundy value 
when the playground is a path in $C_n$ with $k$ vertices. 
Then $f(n)=0$. 

\medskip 

\noindent
There cannot be a playground with one vertex not selected, 
thus we need to leave $f(n-1)$ undefined. 
Next, we have $f(n-2)=1$ since the game ends in the next move when two 
vertices are not selected. 
Subsequently, we find 
\begin{align*}
f(n-3) &= \mex \; \{\;f(n-2),\;f(n)\;\}=2, \\ 
f(n-4) & =\mex \; \{\;f(n-3),\;f(n-2)\;\}=0,  \\
f(n-5) & = \mex \;\{\;f(n-4),\;f(n-3)\;\}=1, \quad\text{\&tc.}
\end{align*}

\medskip 

\noindent
It is now easy to prove, by induction, that, for $i \geq 2$,  
\[f(n-i)=
\begin{cases}
0 & \text{if $i \equiv 1 \pmod 3$}\\
1 & \text{if $i \equiv 2 \pmod 3$}\\
2 & \text{if $i \equiv 0 \pmod 3$.}
\end{cases}\] 
This implies that 
\[ f(1)= 
\begin{cases}
0 & \text{if $n \equiv 2 \pmod 3$}\\
1 & \text{if $n \equiv 0 \pmod 3$}\\
2 & \text{if $n \equiv 1 \pmod 3$.}
\end{cases}\] 

\medskip 

\noindent 
Let $g(n)$ be the Grundy value for the cycle with $n$ vertices. 
Then 
\[g(n) = \mex \;\{\;f(1)\;\}= 
\begin{cases}
1 & \text{if $n \equiv 2 \pmod 3$}\\
0 & \text{otherwise.}
\end{cases}\] 
This proves the theorem. 
\qed\end{proof}

\subsubsection{Paths}

We denote a path with $n$ vertices as $P_n$. 

\begin{theorem}
The first player has a winning strategy in the connected 
$P_3$-game played on a path with $n$ vertices if $n \neq 2$. 
\end{theorem}
\begin{proof}
Let $f(n)$ denote the Grundy value for the path with one of the 
endpoints labeled. We claim that 
\[f(n)=
\begin{cases}
0 & \text{if $n \equiv 1 \pmod 3$}\\
1 & \text{if $n \equiv 2 \pmod 3$}\\
2 & \text{if $n \equiv 0 \pmod 3$.}
\end{cases}
\]
This claim is readily checked via the recurrence
\[f(n)=\mex \;\{\;f(n-1),\;f(n-2)\;\}.\] 

\medskip 

\noindent
Let $g(n)$ be the Grundy value for the path $p(n)$. 
Then we have that 
\[g(n)=\mex \;\{\; f(n), f(n-i)\oplus f(i+1) \;|\; 1 \leq i < n-1\;\}.\]
It follows that, for $n \neq 2$, 
\[g(n)=
\begin{cases}
\mex \;\{\;0,\;3\;\}=1 & \text{if $n \equiv 1 \pmod 3$}\\
\mex \;\{\;0,\;1\;\}=2 & \text{if $n \equiv 2 \pmod 3$}\\
\mex \;\{\;0,\;2\;\}=1 & \text{if $n \equiv 0 \pmod 3$.}
\end{cases}
\]
This proves the theorem. 
\qed\end{proof}

\section{The $P_3$-game on cographs}

A $P_4$ denotes a path with four vertices. 

\begin{definition}
A cograph is a graph without induced $P_4$. 
\end{definition}

One characterization of cographs is that, every induced 
subgraph with at least two vertices is either a join or a 
union of two smaller cographs. 

\begin{theorem}
There exists an algorithm to decide the $P_3$-game in polynomial time  
on cographs. 
\end{theorem}
\begin{proof}
Let $G$ be a cograph. 
When $G$ is disconnected, the game reduces to the sum of 
the games played on the components of $G$. 
By the Sprague-Grundy theorem, the Grundy value of the game is the 
nim-sum of the Grundy values played on the components. 

\medskip 

\noindent 
Assume that $G$ is connected. Then $G$ is the join of two smaller 
cographs $G_1$ and $G_2$, that is, all vertices 
of $G_1$ are adjacent to all vertices of $G_2$.  
The algorithm considers all possible 
playgrounds after both players have made a move. 
First assume that both players chose a vertex of $G_1$. 
Then all vertices of $G_2$ are added to the hull after the second move. 
When $G_2$ has at least two vertices, then the game is over, since 
all vertices of $G_1$ are also in the $P_3$-closure. 

\medskip 

\noindent 
Assume that $G_2$ has only one vertex. Notice that this vertex is 
universal, that is, it is adjacent to all other vertices. 
Let $C_1, \dots, C_t$ be the components of $G_1$. The components that 
contain $x$ and $y$ are subsets of the hull. The components that do not 
contain $x$ nor $y$ are added one by one in every subsequent move. 
It follows that, player one wins the game from this position 
if and only if the number of 
components that do not contain $x$ or $y$ is odd. 
\medskip 

\noindent 
Assume that player one makes his first move in $G_1$ and that 
player two makes his first move in $G_2$. Then the next move 
adds a vertex of $G_1$ or $G_2$ and this reduces the analysis to 
one of the 
previous cases. 

\medskip 

\noindent 
This proves the theorem. 
\qed\end{proof}

\section{The connected $P_3$-game on ladders}

\begin{definition}
A ladder is the Cartesian product of two paths, one of which has 
only one edge. 
\end{definition}
A ladder is denoted as $L_n=P_2 \times P_n$ and it has $2n$ vertices 
and $2(n-1)+n$ edges. It consists of two paths $P_n$ of length $n-1$ and a 
perfect matching. The edges of the matching are called the rungs and the 
two paths $P_n$ are called the ringers, or rails or stiles.  

\begin{theorem}
The first player has a winning strategy for the connected $P_3$-game on a 
ladder $L_n$ if and only if $n \equiv 0 \pmod 6$. 
\end{theorem}
\begin{proof}
First assume that $n \equiv 0 \pmod 6$. Partition the ladder into dominoes, 
that is, $P_2 \times P_3$. To win the game, player one starts in a corner 
of the ladder, that is, a vertex of degree two. It is easy to check that, 
no matter what player two plays, player one always completes the 
first, and then every subsequent domino of the ladder. 

\medskip 

\noindent
Now assume that $n \not\equiv 0 \pmod 6$. Assume player one starts in a 
corner. Then player two chooses a vertex in the other stile, such that 
the remaining rungs can be partitioned into dominoes. Then, when player 
one enters a domino, player two can choose a vertex which encloses 
the domino in the $P_3$-closure. 

\medskip 

\noindent
Assume player one starts in some middle rung.  Player two chooses a 
vertex in the other stile, such that the two numbers of rungs below and above 
the $P_3$-closure are equivalent modulo 3. Each next move adds either 
one or two rungs to the $P_3$-closure. The strategy of player two 
is to keep the numbers of remaining rungs modulo three, below and above the 
$P_3$-closure the same. 
\qed\end{proof}

\section{The connected $P_3$-game on caterpillars}

A caterpillar is a tree that contains a dominating path, that 
is, a path such that every other vertex is connected to a vertex in the 
path.  Alternatively, a caterpillar is a tree without the 
subdivision of a star $K_{1,3}$ as a subgraph. 
The algorithm that decides the connected $P_3$-game for trees simplifies 
a little bit in case the tree is a caterpillar. In this section we 
shortly describe this simplification. 

\bigskip 

The dominating path of the caterpillar is called the backbone and 
we denote it by $P$. Assume that 
$P$ is a path with $n$ vertices, $P \simeq P_n$. 
We assume that the endpoints of the backbone are 
leaves of the caterpillar. The vertices of the backbone are numbered, 
$1,\dots,n$. Let the number of feet adjacent to the point $i$ be 
denoted as $h(i)$. 

\bigskip 

\begin{theorem}
There exists an efficient algorithm to decide if the first player 
has a winning strategy in the connected $P_3$-game played on a 
caterpillar. 
\end{theorem}
\begin{proof}
We analyze all possible positions after two moves have been made. 
First assume that the first and second move are the selection 
of two feet $x$ and $y$ adjacent to the same vertex $i$ in the  $P$. 
In that case, the $P_3$-closure is the $P_3$, $\{x,y,i\}$. 
The remaining game is split into a caterpillar with backbone 
$P[\{1,\dots,i\}]$, one caterpillar with backbone $P[\{i,\dots,n\}]$, 
and $h(i)-2$ edges with $i$ as an endpoint. 
The Grundy value for this position is the nim-sum of the 
games described above, with initially labeled vertex $i$. 
Note that the nim-sum of the $h(i)-2$ edges with labeled vertex $i$ is 
\[\oplus_{\ell =1}^{h(i)-2} \; 1 = 
\begin{cases} 
1 & \text{if $h(i)$ is even}\\
0 & \text{otherwise.}
\end{cases}\] 

\medskip 

\noindent
Assume that the first move is the selection of a foot $x$ adjacent to 
$i$, and the second move is the selection of the vertex $i$. 
The remaining game is split into a game 
on a caterpillar with backbone $P[\{1,\dots,i\}]$, one caterpillar 
with backbone $P[\{i,\dots,n\}]$, and $h(i)-1$ games on edges 
with endpoint $i$. 

\medskip 

\noindent 
Assume that the first move is the selection of a foot $x$ adjacent to 
$i$ and the second move is the selection of the vertex $i+1$. 
The $P_3$-closure is $\{x,i,i+1\}$, and the game is split into 
a game on the caterpillars with backbones $P[\{1,\dots,i\}]$ and 
$P[\{i+1,\dots,n\}]$, the $h(i)-1$ edges consisting of the remaining leaves 
adjacent to $i$ and the 
$h(i+1)$ edges connecting leaves to $i+1$. 

\medskip 

\noindent
The first two moves are the selection of vertices $i$ and $i+1$. 
The game is split into two caterpillars, one with backbone 
$P[\{1,\dots,i\}]$ and one with backbone $P[\{i+1,\dots,n\}]$, 
and $h(i)$ and $h(i+1)$ edges connecting leaves with $i$ and $i+1$ 
respectively. 

\medskip 

\noindent 
The remaining case is where two vertices $i$ and $i+2$ of the 
backbone are selected. The game is split, similar as to that 
described above. 

\medskip 

\noindent
It follows that the game can be decided by an algorithm that performs 
dynamic programming on caterpillars with backbones that are subpaths 
of $P$ with one end a labeled vertex of $P$ and the other end the point 
$1$ or $n$. This algorithm can be implemented to run in $O(n)$ time; 
precisely speaking, it is the number of these subpaths of $P$. 

\medskip 

\noindent
This proves the theorem. 
\qed\end{proof}

\section{The connected $P_3$-game on chordal graphs}

\begin{definition}
A graph is chordal if it has no induced cycle of length more than three. 
\end{definition}

The following lemma is easily checked. For a proof see, eg, Centeno 
et al. 

\begin{lemma}
If $x$ and $y$ are two vertices at distance at most two in a biconnected 
chordal graph $G$, then 
\[\sigma(\{x,y\})=V(G).\]
\end{lemma}

\begin{theorem}
There exists a polynomial-time algorithm to decide the connected 
$P_3$-game on connected chordal graphs. 
\end{theorem}
\begin{proof}
The algorithm considers all possible initial playgrounds after the 
first two moves have been made. If the graph is not connected, 
then the vertices that are not in the initial playground, are 
partitioned into components for which the neighborhood 
is a cutvertex. Consider the components, rooted at the cutvertices. 
By the Sprague-Grundy theorem, the Grundy value of the playground is the 
nim-sum of the games played at these components. It follows that this is 
computable in polynomial time by dynamic programming on the biconnected 
components.\footnote{For the definition and properties of biconnected component, we refer the readers to {\tt https://en.wikipedia.org/wiki/Biconnected\_component}.} 
\qed\end{proof}

\section{The connected $P_3$-game on cacti}

In this section we prove that the connected $P_3$-game on cacti is decidable in 
polynomial time. 

\begin{definition}
A cactus is a graph without the diamond as a minor. 
\end{definition}
Equivalently, a graph is a cactus if every edge is a subset of at 
most one cycle of the graph. Also, a graph is a cactus if every biconnected 
component is an edge or a cycle. 

\bigskip 

\begin{theorem}
There exists a polynomial-time algorithm to decide the connected 
$P_3$-game on cacti. 
\end{theorem}
\begin{proof}
We show that the Grundy value is computable in polynomial time. 

\medskip 

\noindent 
The algorithm tries all possible vertices as a first move for player one. 
Consider a vertex $x$ as a first move. Assume that $x$ is in a cycle and 
let $R$ be a cycle containing $x$. 
First assume that $x$ is a cutvertex. 
Consider the components of $G-x$ augmented with the vertex $x$. 
By the Sprague-Grundy theorem, the Grundy value is the nim-sum 
of the Grundy values of the games played on the augmented components. 

\medskip 

\noindent
For the augmented component that contains $R$, the children of 
the playground $\{x\}$ are the playgrounds that are edges that have $x$ 
as an endpoint and those that are $P_3$'s with $x$ as an endpoint. 
For each of these playgrounds, if it contains a cutvertex $y$, the game 
is split into components of $G-y$ augmented with $y$. The Grundy value is 
the nim-sum of these games. For each of the components that do not contain 
vertices of $R \setminus \{y\}$, the algorithm recursively computes the 
Grundy value. Notice that the initial playground for each component 
is either a single vertex $y$, or an edge incident with $y$.  
When it is an edge, 
the playground extends greedily into maximal biconnected chordal subgraph. 

\medskip 

\noindent 
There are only $O(|R|^2)$ different playgrounds induced on $R$.    
For each playground $r$ on $R$ the algorithm computes the $\mex$ value 
from its children. If a child includes a cutvertex, the game is split. For 
the augmented components that share a cutvertex with $R$, with a playground 
that is either the cutvertex, or an edge incident with the cutvertex, 
the Grundy value 
is computed recursively. 
\qed\end{proof}

\section{The connected $P_3$-game on outerplanar graphs}

A graph is outerplanar if it has a plane embedding with all vertices 
lying on the outerface. 
Alternatively, outerplanar graphs are defined as follows. 

\begin{definition}
A graph is outerplanar if it does not contain $K_4$ or $K_{2,3}$ as a minor.
\end{definition}

\bigskip 

\begin{theorem}
There exists a polynomial-time algorithm that decides the connected 
$P_3$-game on outerplanar graphs. 
\end{theorem}
\begin{proof}
The proof is similar to that for the cacti. 
For ease of description, we 
assume that $G$ is a biconnected outerplanar graph. 
The algorithm that we describe extends in an obvious manner 
for cases when $G$ contains cutvertices. 
When $G$ is biconnected, it forms a `tree of cycles,' which can 
be defined recursively as follows.  Any cycle is a biconnected 
outerplanar graph. A biconnected outerplanar graph $G^{\prime}$, 
with an outerface $O$, and a cycle $C$, a new biconnected outerplanar 
graph is formed by identifying the endpoints of an edge with the 
endpoints of an edge in $O$. This recursively defines all biconnected 
outerplanar graphs. 

\medskip 

\noindent
All minimal separators are edges and for each edge $e$, $G-e$ contains 
exactly two components. We call the components with the edge $e$ added 
to it, the augmented components at the edge $e$. 

\medskip 

\noindent 
For each separator $e$, and for each augmented component $C$ at $e$, 
the algorithm recursively computes the Grundy value for the playground 
that consists of $e$, and for the playground consisting of $e$ plus 
one additional vertex adjacent to an endpoint of $e$. 
When the cycle incident with $e$ is a triangle, the playground 
greddily extends. 

\medskip 

\noindent 
To process $G$, the algorithm tries all vertices $x$ as an initial 
playground $p$. The children of $p$ are edges incident with $x$ and 
$P_3$'s with $x$ as an endpoint. When a child contains a minimal 
separator, the game splits into two games played on the augmented 
components. In that case, the Grundy value is the nim-sum of the two 
games on the augmented components. 

\medskip 

\noindent 
Consider a playground $p$ that contains only edges of the outerface $O$. 
Then $p$ is a simple path. 
The children are all extensions of $p$ with an edge or a $P_3$. An extension 
$q$ that contains a minimal separator splits the game, and the Grundy value 
is the nim-sum of the two augmented components, each with the 
induced paths 
$q_1$ and $q_2$ as a playground. 

\medskip 

\noindent 
For the children $q$ that extend $p$ with edges along $O$, the 
Grundy value is computed by dynamic programming. Finally, for the 
playground $p$, the Grundy value is computed as the $\mex$ function 
of the Grundy values of its children. 

\medskip 

\noindent 
This proves the theorem. 
\qed\end{proof}

\section{The connected $P_3$-game on circular arc graphs}

Gavril initiated the research on circular-arc graphs. 
These graphs are the intersection graphs 
of arcs on a circle. McConnell showed that this class 
can be recognized in linear time.  Whilst intervals 
on the real line satisfy the Helly property, this is 
no longer true for circular arcs on a circle. That is, 
there could be a triangle in the graph without 
any point on the circle that is in all three arcs. 

\bigskip 

Gavril defines a  Helly circular-arc graph as a graph 
for which the clique matrix has the circular 1s property 
for columns. Also the Helly circular-arc graphs are recognizable 
in linear time. 

\begin{definition}
A graph is a Helly circular-arc graph if it has a 
circular maximal clique arrangement.
\end{definition}

\begin{theorem}
There exists a polynomial-time algorithm 
to decide the connected $P_3$-game on Helly circular-arc graphs.
\end{theorem}
\begin{proof}
Obviously, we may assume that the graph is connected. 

\medskip 

\noindent 
To reduce the circular problem to a linear one, the algorithm considers 
all possible first two moves.  

\medskip 

\noindent
When the first two vertices are selected, 
these two vertices and their common neighbors are labeled.  
The closure recursively 
labels all vertices in cliques that contain 
two previously labeled vertices.

\medskip 

\noindent
There is at least one maximal clique in which all vertices are labeled, and so, 
the graph can be regarded as an interval graph, i.e. a graph that has an 
linear order for its maximal cliques.

\medskip 

\noindent
There may exist one or two vertices that appear in  
(consecutive) maximal cliques on the lefthand side of the linear 
order and in the maximal cliques on the righthand side. 
These vertices do not appear in consecutive maximal cliques, 
but, since they are labeled, they can be ignored in subsequent procedures.

\medskip 

\noindent
The vertices that are not in the initial playground are 
partitioned into two types of components:
\begin{enumerate}[\rm (1)]
\item the components that have a neighborhood consisting of a 
single labeled cutvertex, 
and 
\item the components of which the neighborhood 
consists of two labeled cutvertices that border the 
consecutive cliques of the component. 
\end{enumerate}

\medskip 

\noindent 
For a component of the first type, 
it takes exactly one additional move to label all its vertices.
Components of the second type  
are handled recursively by considering 
all possible extra moves, 
rooted at distance at most two from either the cutvertex on 
the righthand side or the lefthand side.

\medskip 

\noindent
This shows that the Grundy values of components can be computed recursively. 
By dynamic programming, the program calculates the 
Grundy value of any sequence of 
components.  
The Grundy value for the circular arc graph is computed by considering 
all possible first two moves. 
\qed
\end{proof}


\bigskip

With a different analysis, the algorithm for Helly circular-arc graphs would accommodate general circular-arc graphs.
Instead of considering max cliques when analyzing the first two moves, we consider \emph{Helly cliques}, a maximal set of arcs overlapping in one point of the circle.
After the selection of the first two vertices, there is at least one Helly clique in which all vertices (arcs) are labeled.
For any such labeled Helly clique, it cuts the circle with respect to the labeling so that the circle becomes an arc $A$.
The rest of the input arcs (on $A$) can be regarded as intervals on a line.
Thus the graph of the rest of input becomes an interval graph and the algorithm for Helly circular-arc graphs proceeds.

\begin{theorem}
There exists a polynomial-time algorithm to decide the connected $P_3$-game on general circular-arc graphs.
\end{theorem}

\section{Concluding remark}

We are aware of only very few problems that are solvable in polynomial 
time on outerplanar graphs, and that resist efficient algorithms for 
graphs of treewidth two. At the moment we do not have an efficient 
algorithm to decide the connected $P_3$-game for graphs of treewidth two. 

\bigskip 

We are interested in the connected $P_3$-game on convex geometries. 
These are alignments satisfying the anti-exchange 
property. We are primarily interested in the question  
under which additional conditions the connected $P_3$-game becomes 
decidable on convex geometries  
(see~\cite{kn:chvatal,kn:edelman,kn:grier}).


\begin{thebibliography}{99}

\bibitem{kn:centeno}Centeno,~C., M.~Dourado, L.~Penso, D.~Rautenbach 
and J.~Szwarcfiter, 
Irreversible conversion of graphs, 
{\em Theoretical Computer Science\/} {\bf 412} (2011), 
pp.~3693--3700. 

\bibitem{kn:chvatal}Chv\'atal,~V., Antimatroids, betweenness, convexity. 
In (W.~Cook, L.~Lov\'asz and J.~Vygen eds.) {\em Research trends in 
combinatorial optimization\/}, Bonn 2008, Springer-Verlag, Berlin, pp.~57--74. 

\bibitem{kn:duchet}Duchet,~P., 
Convex sets in graphs, II. Minimal path convexity, 
{\em Journal of Combinatorial Theory, Series B\/}, {\bf 44} (1988), pp.~307--316. 

\bibitem{kn:edelman}Edelman,~P. and R. Jamison, 
The theory of convex geometries, 
{\em Geometriae Dedicata\/} {\bf 19} (1985), pp.~247--270. 

\bibitem{kn:gavril}Gavril,~F., 
Algorithms on circular-arc graphs, 
{\em Networks\/} {\bf 4} (1974), pp.~357--369. 

\bibitem{kn:grier}Grier,~D., 
Deciding the winner of an arbitrary finite poset game is 
$\PSPACE$-complete, 
{\em Proceedings ICALP'13\/}, Springer-Verlag, LNCS 7965 (2013), pp.~497--503. 

\bibitem{kn:grundy}Grundy,~P., 
Mathematics and games, 
{\em Eureka\/} {\bf 2} (1939), pp.~6--8. 

\bibitem{kn:joeris}Joeris,~B., M.~Lin, R.~McConnell, 
J.~Spinrad and J.~Szwarcfiter, 
Linear-time recognition of Helly circular-arc models and 
graphs, 
{\em Algorithmica\/} {\bf 59} (2009), pp.~215--239. 

\bibitem{kn:kloks}Kloks,~T. and Y.~Wang, 
{\em Advances in graph algorithms\/}. Manuscript on ViXrA:1409.0165, 2014. 

\bibitem{kn:korte}Korte,~B., L.~Lov\'asz and R.~Schrader, 
{\em Greedoids\/}, Springer-Verlag, Berlin, 1990. 

\bibitem{kn:mcconnell}McConnell,~R., 
Linear-time recognition of circular-arc graphs, 
{\em Algorithmica\/} {\bf 37} (2013), pp.~93--147. 

\bibitem{kn:nienhuys}Nienhuys,~J., 
Graphs and Games. In (L.~Hung and T.~Kloks eds.) 
{\em De~Bruijn's Combinatorics --- Classroom notes\/}. 
Manuscript on ViXrA:1208.0223, 2012. 

\end{thebibliography}
\end{document}